\documentclass[journal,10pt,a4paper,twocolumn,twoside]{IEEEtran}

\usepackage{amsthm}

\usepackage{graphicx}
\usepackage{subcaption}
\usepackage{amsmath}
\usepackage{booktabs}

\usepackage{xcolor}
\usepackage{amssymb}
\usepackage{amsmath,cases}
\usepackage{dsfont}
\usepackage{physics}

\usepackage{hyperref}
\usepackage{cite}
\usepackage{newpxtext,newpxmath}
\usepackage{acronym}

\definecolor{ocre}{HTML}{800000}
\definecolor{green}{RGB}{0, 128, 0}
\definecolor{sky}{HTML}{C6D9F1}
\definecolor{skybox}{HTML}{5F86B3}

\usepackage{tikz}
\usetikzlibrary{quantikz}
\usetikzlibrary{decorations.text, arrows.meta,calc,shadows.blur,shadings,intersections}
\usetikzlibrary{trees}

\usepackage{mathtools}
\usepackage{subcaption}
\usepackage{cite}
\usepackage{amsmath,amssymb,amsfonts}
\usepackage{algorithmic}
\usepackage{graphicx}
\usepackage{textcomp}
\usepackage{xcolor}

\usepackage{hyperref}

\newtheorem{theorem}{\textbf{Theorem}}

\newtheorem{lemma}{\textbf{Lemma}}
\newtheorem{corollary}{\textbf{Corollary}}

\begin{document}
\title{Crosstalk-Resilient Quantum MIMO for Scalable Quantum Communications}
	\author{Seid Koudia, Symeon Chatzinotas\\
	
	\thanks{
 The authors are with the Interdisciplinary Centre for Security, Reliability, and Trust (SnT), Luxembourg, L-1855 Luxembourg.

 }
    }

	\maketitle
\begin{abstract}
    We address the challenge of crosstalk in quantum multiplexing—an obstacle to scaling throughput in quantum communication networks. Crosstalk arises when physically coupled quantum modes interfere, degrading signal fidelity. We propose a mitigation strategy based on encoding discrete-variable (DV) quantum information into continuous-variable (CV) bosonic modes using Gottesman-Kitaev-Preskill (GKP) codes. By analyzing the effect of mode-mixing interference, we show that under specific noise strength conditions, the interaction can be absorbed into a gauge subsystem that leaves the logical content intact. We provide rigorous conditions for perfect transmission in the ideal case, derive the structure of the output codes and prove the existence of a gauge-fixing decoder enabling recovery of the logical information. Numerical simulations under displacement Gaussian noise illustrate the fidelity behavior and rate-fidelity tradeoff. Our results establish a coding-theoretic foundation for crosstalk-resilient multiplexing in quantum networks.
\end{abstract}
\section{Introduction}
The vision of a scalable quantum internet \cite{kimble2008quantum} depends critically on the ability to establish high-fidelity end-to-end entanglement across large, distributed quantum networks \cite{wehner2018quantum, koudia2023quantum}. Achieving this goal requires both increasing the \emph{throughput} of quantum links and maintaining the \emph{fidelity} of transmitted quantum states \cite{vischi2023simulating}. Multiplexing---the simultaneous use of multiple channels over a shared resource---has emerged as a key strategy for boosting quantum communication rates \cite{wengerowsky2018entanglement, seshadreesan2020continuous}. Meanwhile, diversity techniques \cite{oleynik2025diversity, ur2025mimo}---the use of redundant or parallel transmissions to improve reliability---offer a path toward robust state transfer under noise \cite{oleynik2025diversity, koudia2024spatial,wang2025exploiting}. These approaches are foundational for advanced quantum communication paradigms, including Quantum Multiple-Input Multiple-Output (MIMO) systems \cite{ur2025mimo,oleynik2025diversity}, which leverage multiple spatial modes to enhance data rates and robustness. However, the increased complexity of such systems makes them inherently more susceptible to crosstalk. Both multiplexing and diversity are fundamentally challenged by the problem of \emph{crosstalk}: interference between quantum modes caused by physical coupling, such as mode mixing in photonic hardware \cite{vischi2023simulating, hoch2025quantum, koudia2024physical}. Crosstalk can degrade the fidelity of quantum transmission and limit the effectiveness of network-level parallelism.

\color{black}
In quantum information science, DV and CV systems offer complementary strengths. DV systems, such as qubits, are naturally suited to logical operations and are standard in many communication protocols. However, they are highly susceptible to noise and crosstalk, especially in multiplexed settings. On the other hand, CV systems—particularly bosonic modes—offer a rich structure for encoding and robust error correction but often lack native digital logic. The interplay between these paradigms opens new possibilities: by encoding DV quantum information into CV modes via GKP codes, one can combine the discrete control and processing advantages of DV systems with the error resilience and resource efficiency of CV encodings. In this work, we propose to mitigate crosstalk in multiplexed quantum networks by encoding DV quantum states into CV bosonic modes using Gottesman-Kitael-Preskill (GKP) codes \cite{gottesman2001encoding}. GKP states are a powerful family of bosonic quantum error-correcting codes that embed DV logical information into structured grids in phase space \cite{conrad2022gottesman}. They have found broad application in quantum technologies, ranging from fault-tolerant quantum computing \cite{grimsmo2021quantum, wang2024passive, schmidt2022quantum, larsen2025integrated} and metrology \cite{labarca2024quantum} to recent advances in quantum communication \cite{schmidt2024error}. In particular, GKP codes are promising candidates for photonic quantum networks, including all-optical quantum repeater schemes \cite{fukui2021all}. These platforms inherently rely on CV hardware (e.g., squeezed light, beamsplitters, homodyne detection), making GKP encodings naturally compatible with physical-layer implementations.
\color{black}

In this paper we investigate whether encoding DV states into GKP modes enables crosstalk-resilient multiplexing---under specific, quantifiable conditions on the crosstalk strength. We show that when the crosstalk strength lies at specific values, and the GKP stabilizers are appropriately scaled, the interference caused by mode coupling can be absorbed into a gauge subsystem. As a result, the logical information is preserved and can be recovered deterministically. We analyze this effect rigorously, characterizing the structure of the output state, the code conditions for perfect transmission, and the emergent gauge symmetry. Furthermore, we show how this structure can be exploited through gauge-fixing procedures, enabling reliable logical recovery. Our work establishes a foundation for crosstalk-aware encoding in quantum networks based on CV hardware and opens up new directions in entanglement distribution through structured code design, thereby contributing to the development of robust Quantum MIMO communications.

\color{black}
The structure of the paper is as follows. In Section.~\ref{preliminaries}, we introduce the GKP code formalism, define our notation, and explain how DV quantum information is encoded into CV bosonic modes. In Section.~\ref{model}, we present the system model, including the crosstalk mechanism arising from mode mixing via beam splitter-induced, and formulate the problem of logical information preservation under such interference. Section.~\ref{results} contains the core theoretical results: we derive precise conditions under which crosstalk-induced displacements can be absorbed into the stabilizer structure of GKP codes, ensuring perfect logical transmission. We provide lattice-matching criteria, and prove a theorem identifying rational transmissivity values that allow for crosstalk-mitigated multiplexing. In Section.~\ref{output}, we analyze the structure of the output state, showing its decomposition into logical and gauge subsystems. We derive the form of the gauge entanglement, propose a gauge-fixing procedure, and construct a decoder using modular arithmetic and Clifford operations. This section also includes numerical simulations illustrating fidelity under Gaussian noise and trade-offs between code dimension and performance. Finally, in Section.~\ref{conclusions}, we conclude by summarizing our contributions and outlining future directions for applying structured GKP codes in scalable quantum networking.
\color{black}

\color{black}
\section{GKP States and logical encoding}
\label{preliminaries}
In this section, we introduce our notations and the GKP states. For a bosonic mode with creation and annihilation operators \( \hat{a} \) and \( \hat{a}^\dagger \) satisfying \( [\hat{a}, \hat{a}^\dagger] = 1 \), the quadrature operators are
\[
\hat{q} = \frac{\hat{a} + \hat{a}^\dagger}{\sqrt{2}}, \quad \hat{p} = \frac{-i(\hat{a} - \hat{a}^\dagger)}{\sqrt{2}}
\]
where we choose the convention \( \hbar = 1 \). The translation operator in the phase space is defined as
\[
\hat{T}(\boldsymbol{u}) \equiv \exp[i(u_p \hat{q} - u_q \hat{p})],
\]
where \( \boldsymbol{u} = (u_q, u_p) \in \mathbb{R}^2 \) is the displacement vector. Notice that \( \hat{T}(\boldsymbol{u}) = \hat{D}((u_q + i u_p)/\sqrt{2}) \), where \( \hat{D}(\alpha) = \exp(\alpha \hat{a}^\dagger - \alpha^* \hat{a}) \) is the usual displacement operator. The commutation relation is
\[
\hat{T}(\boldsymbol{u}) \hat{T}(\boldsymbol{v}) = e^{-i \omega(\boldsymbol{u}, \boldsymbol{v})} \hat{T}(\boldsymbol{v}) \hat{T}(\boldsymbol{u}),
\]
where the phase \( \omega(\boldsymbol{u}, \boldsymbol{v}) = u_q v_p - u_p v_q = \det\begin{pmatrix} u_q & u_p \\ v_q & v_p \end{pmatrix} \) is the oriented area of the parallelogram spanned by vectors \( \boldsymbol{u} \) and \( \boldsymbol{v} \).

GKP states are defined as states that are invariant under certain phase space displacements. Given two displacement operators \( \hat{S}_X = \hat{T}(\boldsymbol{u}) \) and \( \hat{S}_Z = \hat{T}(\boldsymbol{v}) \), the states that are invariant under the stabilizers \( \hat{S}_X \) and \( \hat{S}_Z \) form a \( d \)-dimensional GKP code space \cite{conrad2022gottesman, mensen2021phase}:
\[
\mathcal{C}(\boldsymbol{u}, \boldsymbol{v}) \equiv \left\{ |\psi\rangle \,:\, \hat{S}_X |\psi\rangle = \hat{S}_Z |\psi\rangle = |\psi\rangle \right\}
\]
if the area \( \omega(\boldsymbol{u}, \boldsymbol{v}) = 2\pi d \) and \( d \) is an integer. We can equivalently denote a GKP code as \( \mathcal{C}_{d, \mathbf{S}} \) where \( \mathbf{S} \) is the normalized basis matrix for the GKP lattice defined as
\[
\sqrt{2\pi d} \, \mathbf{S} = \begin{pmatrix} u_q & u_p \\ v_q & v_p \end{pmatrix} \equiv (\boldsymbol{u}, \boldsymbol{v}).
\]
It is easy to see that \( \mathbf{S} \) is a \( 2 \times 2 \) symplectic matrix since \( \det(\mathbf{S}) = 1 \) \footnote{\color{black}The condition of having $det(S)=1$ is sufficient in this case, because by construction $S$ preserves the lattice area.}. Throughout the paper, we will interchangeably use the notation \( \mathcal{C}(\boldsymbol{u}, \boldsymbol{v}) \) or \( \mathcal{C}_{d, \mathbf{S}} \) to represent a GKP code based on convenience.

The logical operators for a GKP code \( \mathcal{C}(\boldsymbol{u}, \boldsymbol{v}) \) are \( \hat{X} = \hat{T}(\boldsymbol{u}/d) \) and \( \hat{Z} = \hat{T}(\boldsymbol{v}/d) \), satisfying \( \hat{X} \hat{Z} = e^{2\pi i/d} \hat{Z} \hat{X} \). We choose the eigenstates \( |\mu\rangle, \mu = 0, \dots, d-1 \) of \( \hat{Z} \) as the basis states of the GKP code \( \mathcal{C}(\boldsymbol{u}, \boldsymbol{v}) \), where
\[
\hat{Z}|\mu\rangle = e^{2\pi i \mu/d} |\mu\rangle, \quad \hat{X}|\mu\rangle = |\mu + 1 \mod d\rangle.
\]

For a square lattice \footnote{\color{black}For illustrative purposes, we employ a square lattice; any other lattice with the same phase space area is equivalent via a rotation, offering no inherent advantage.} GKP code \( \mathcal{C}_{d, \mathbf{S}} \) where \( \mathbf{S} = \mathbb{I}_2 \) is the \( 2 \times 2 \) identity matrix, the basis states in the position space are
\[
|\mu\rangle = \sum_{k \in \mathbb{Z}} |\hat{q} = \sqrt{2\pi d}(k + \mu/d) \rangle, \quad \mu = 0, \dots, d-1.
\]
The logical operators for a square lattice GKP code \( \mathcal{C}_{d, \mathbf{S}} \) are
\[
\hat{X} = \exp\left( i \sqrt{2\pi/d} \hat{p} \right), \quad \hat{Z} = \exp\left( i \sqrt{2\pi/d} \hat{q} \right).
\]

Finally, when \( d = 1 \), the GKP state \( |\mu = 0\rangle \) is called a qunaught state since it does not carry any quantum information.

\section{Transmission Model and Problem Statement}
\label{model}
We now consider two GKP-encoded modes, labeled $A$ and $B$, each carrying logical information via lattices defined by vectors $(\boldsymbol{f}_A, \boldsymbol{g}_A)$ and $(\boldsymbol{f}_B, \boldsymbol{g}_B)$ respectively. The logical operators for each mode are:
\[
\hat{X}_i = T\left( \frac{\boldsymbol{f}_i}{d_i} \right), \quad \hat{Z}_i = T\left( \frac{\boldsymbol{g}_i}{d_i} \right), \quad i \in \{A, B\}.
\]
\color{black}
These operators generate the logical Pauli group for an embedded $d_i$-dimensional DV system within each CV  bosonic mode. The eigenstates of $\hat{Z}_i$ serve as the computational basis $\{|\mu\rangle_i\}_{\mu=0}^{d_i-1}$, which are localized along a one-dimensional lattice in phase space. The logical operator $\hat{X}_i$ acts as a shift between these basis states, while $\hat{Z}_i$ applies a phase. Together, $\hat{X}_i$ and $\hat{Z}_i$ satisfy the Weyl commutation relation:
\[
\hat{Z}_i \hat{X}_i = e^{2\pi i / d_i} \hat{X}_i \hat{Z}_i,
\]
mirroring the algebra of qudit Pauli operators. In this way, the GKP code enables encoding of DV logical states (such as qubits or qudits) into structured CV states. 
\color{black}
\subsection{Crosstalk via node mixing}

\begin{figure}[t]
    \centering
    \includegraphics[width=1\linewidth]{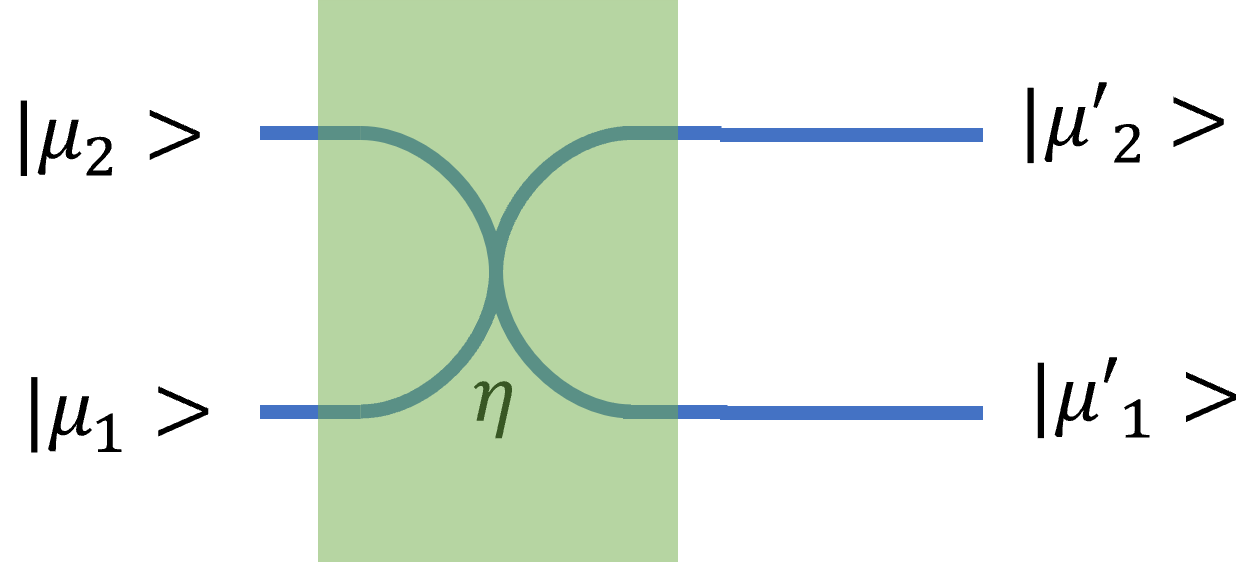}
    \caption{Crosstalk between two GKP-encoded modes via mode mixing beam splitter interaction}
    \label{fig:crosstalk}
\end{figure}
\color{black}
We model the crosstalk arising from mode mixing interfecrence by a beam splitter $\hat{U}_\eta$, jointly transforming the modes, while  $\eta$ relflecting the strength of mode mixing is considered to be stochastic. The beam splitter acts as a linear symplectic transformation:

\begin{equation}
    \hat{U}_\eta |q_A\rangle |q_B\rangle = |\sqrt{\eta} q_A + \sqrt{1 - \eta} q_B\rangle \otimes |\sqrt{\eta} q_B - \sqrt{1 - \eta} q_A\rangle
\end{equation}

where  $\eta$ is modeled by a log-normal distribution, characterized by the probability density function:

\begin{equation}
    f_\eta(\eta) = \frac{1}{\eta \sigma_c \sqrt{2\pi}} \exp\left(-\frac{(\ln \eta - \mu)^2}{2\sigma_c^2}\right)
\end{equation}
with  \( \mu \) being the mean of \( \ln \eta \), and \( \sigma_c \) is its standard deviation reflecting the strength of fading.
This operation introduces undesired correlation between the two modes $A$ and $B$ reflected in the displacement, which potentially disrupts independent logical recovery.
\color{black}
\subsubsection{Problem Statement}
\color{black}
We aim to determine conditions under which logical information in both GKP-encoded modes is preserved after the crosstalk transformation $\hat{U}_\eta$. Ideally, we require:
\begin{equation}
   D(\eta) \hat{U}_\eta |\lambda_A\rangle |\lambda_B\rangle = |\lambda_A\rangle |\lambda_B\rangle
\end{equation}
with $D(\eta)$ would be a decoder depending on a priori knowledge of the average behaviour of the stochastic $\eta$,
or more generally, that the output remains within the original GKP code spaces up to correctable displacements. The main question is: for which pairs of lattices $(\boldsymbol{f}_A, \boldsymbol{g}_A)$ and $(\boldsymbol{f}_B, \boldsymbol{g}_B)$ is the logical subspace preserved under the mode-mixing crosstalk transformation $\hat{U}_\eta$? Equivalently, we ask: under what conditions does the joint stabilizer group transform into an equivalent stabilizer group (modulo correctable errors), such that the logical content in both modes is unaffected by crosstalk?

More precisely, let each mode’s GKP code be defined by stabilizer lattices $\Lambda_A = \mathrm{Span}_{\mathbb{Z}}\{\boldsymbol{f}_A, \boldsymbol{g}_A\}$ and $\Lambda_B = \mathrm{Span}_{\mathbb{Z}}\{\boldsymbol{f}_B, \boldsymbol{g}_B\}$. After undergoing crosstalk, these are mapped to new lattices $\tilde{\Lambda}_A$, $\tilde{\Lambda}_B$ with basis vectors $(\tilde{\boldsymbol{f}}_A, \tilde{\boldsymbol{g}}_A)$ and $(\tilde{\boldsymbol{f}}_B, \tilde{\boldsymbol{g}}_B)$, respectively. Logical preservation (up to correctable displacement) requires that the original stabilizer lattice be contained in the transformed one:
\[
\Lambda_A \subseteq \tilde{\Lambda}_A, \quad \Lambda_B \subseteq \tilde{\Lambda}_B,
\]
and that any displacement incurred by the transformation lies within the correctable region—e.g., the Voronoi cell \cite{conrad2022gottesman}—of the respective new lattices. Under these conditions, the logical subspace is preserved and can be recovered by standard GKP error correction and appropriate decoding.
\color{black}

\section{Results: Crosstalk-Mitigated Multiplexing via GKP Encodings}
\label{results}
In discrete-variable (DV) quantum information processing, crosstalk between spatial or spectral modes can be a major source of noise, \textcolor{black}{if left untreated}. When two modes undergo a mode mixing, information encoded in one mode will generally affect the other, leading to undesired logical errors. In this section, we show how such coupling—modeled as a passive linear interaction—can be made harmless by encoding the DV information into GKP code states. This allows us to multiplex quantum information in a crosstalk-robust manner.

\subsection{The mode mixing Crosstalk}

In order to understand the effect of crosstalk on stabilizers of each mode, we should understand how generic displacement operators transform under this mode mixing model. The reason is that stabilizers of GKP codes are given in terms of displacement operators as seen in Sec.~\ref{preliminaries}. Under beamsplitter evolution, the product of displacements, each acting on its corresponding mode, transforms as:
\begin{equation}
\hat{U}_\eta T_1(\alpha) T_2(\beta) \hat{U}_\eta^\dagger = T_1\left(\sqrt{\eta}\alpha + \sqrt{1 - \eta}\beta\right)
T_2\left(\sqrt{\eta}\beta - \sqrt{1 - \eta}\alpha\right).
\label{eq:beamsplitter_displacement}
\end{equation}
This reveals that displacements applied to one mode result in induced displacements on the other, as expected from a crosstalk mechanism. As a matter of fact, our goal is to find a way to absorb these induced displacements into the stabilizers of a code, such that the logical information remains unchanged.

\subsection{Mitigating Crosstalk via Stabilizer Embedding}

Suppose we wish to transmit logical states encoded in mode 1 and mode 2, undergoing crosstalk, while both modes are intialized  initialized in a GKP codewords. In order to prevent the induced displacement from corrupting both modes, we aim to ensure that this displacement falls within the code’s stabilizer lattice.

\paragraph*{Observation:} If the second mode is initialized in a GKP code with stabilizers \(\mathcal{S}_2 = \langle T_2(\mathbf{u}_2), T_2(\mathbf{v}_2) \rangle\), then any displacement of the form \(T_2(\gamma)\) with \(\gamma \in \Lambda_2\) (the stabilizer lattice) acts trivially on the code space.\\
Now suppose we apply a displacement \(T_1(\alpha)\) to mode 1. From Eq.~\eqref{eq:beamsplitter_displacement}, we notice that
to ensure the induced displacement on mode 2 does not affect the output we require that:
\begin{equation}
\sqrt{1 - \eta} \, \alpha = \sqrt{\eta} \beta \Rightarrow \quad \alpha = \frac{\sqrt{\eta}}{\sqrt{1 - \eta}} \mathbf{\beta}, \quad \mathbf{\beta} \in \Lambda_2.
\label{eq:alpha_requirement}
\end{equation}
The key point is that \(\beta\) was arbitrary in \(\Lambda_2\), so if we define $\beta$ to be one of the lattice vectors of the GKP code of mode 2 $(\mathbf{u}_2,\mathbf{v}_2)$ then:
\begin{equation}
\tilde{\mathbf{u}}_1 = \sqrt{\frac{\eta}{1 - \eta}} \, \mathbf{u}_2, \quad
\tilde{\mathbf{v}}_1 = \sqrt{\frac{\eta}{1 - \eta}} \, \mathbf{v}_2,
\end{equation}
then \(\tilde{\mathbf{u}}_1, \tilde{\mathbf{v}}_1\) generate a new lattice in mode 1, and any displacement in the group:
\begin{equation}
\Lambda(\tilde{\mathbf{u}}_1, \tilde{\mathbf{v}}_1) := \{ T_1(s \tilde{\mathbf{u}}_1 + t \tilde{\mathbf{v}}_1) \,|\, s, t \in \mathbb{Z} \}
\label{lattice}
\end{equation}
can be realized through the beamsplitter interaction from mode 2.


Subsequently,  investigate the effect of this transmofrmation of the symplectic form of the generators \(\tilde{\mathbf{u}}_1, \tilde{\mathbf{v}}_1\), assuming \(\omega(\mathbf{u}_2, \mathbf{v}_2) = 2\pi d_2\). Then:
\[
\omega(\tilde{\mathbf{u}}_1, \tilde{\mathbf{v}}_1)
= \left( \frac{\eta}{1 - \eta} \right) \omega(\mathbf{u}_2, \mathbf{v}_2)
= \frac{\eta}{1 - \eta} \cdot 2\pi d_2.
\]
In order to define a GKP code on mode 1 with stabilizer area \(\omega = 2\pi d_1\), the new lattice vectors \(\tilde{\mathbf{u}}_1, \tilde{\mathbf{v}}_1\) must satisfy:
\[
\omega(\tilde{\mathbf{u}}_1, \tilde{\mathbf{v}}_1) = \frac{2\pi q}{p d_1}, \quad \text{for some } m, k \in \mathbb{Z}.
\]

Combining the expressions yields:
\begin{equation}
    \frac{\eta}{1 - \eta} \cdot 2\pi d_2 = \frac{2\pi q}{p d_1}
\quad \Rightarrow \quad
\eta = \frac{q}{q + p d_1 d_2}
\label{cond}
\end{equation}
This provides a discrete set of values for \(\eta\) where perfect GKP logical transmission is possible. We summarize this result below. 

\begin{theorem}[Perfect Transmission Under Crosstalk]
Let two GKP codes \(\mathcal{C}_{S_1, d_1}, \mathcal{C}_{S_2,d_2}\) with stabilizer areas \(2\pi d_1\) and \(2\pi d_2\) be encoded in two bosonic modes undergoing crosstalk \(\hat{U}_\eta\). Then there exists a choice of stabilizers for both codes such that the logical information is perfectly preserved if and only if:
\[
\eta = \frac{q}{q + p d_1 d_2}, \quad q, p \in \mathbb{Z}.
\]
\end{theorem}
The above implies a matching condition between the stabilizer lattices. In particular, if \(S_1\) and \(S_2\) are generator matrices for the two lattices, then they must be related by:
\[
S_2 = 
\begin{pmatrix}
\sqrt{\frac{p_2 q_1}{p_1 q_2}} & 0 \\
0 & \sqrt{\frac{q_1 p_2}{p_2 q_1}}
\end{pmatrix}
S_1,
\]
under the constraint \(q = q_1 q_2\), \(p = p_1 p_2\). This guarantees that each logical mode’s stabilizers are compatible with the crosstalk pattern imposed by mode mixing. \color{black} One may ask whether adapting the lattices to satisfy the matching condition sacrifices GKP code quality. Importantly, \emph{the total lattice area is preserved} in the symplectic sense. However, individual codes may require rescaling — effectively increasing the physical extent of their stabilizer generators — to satisfy the rational compatibility constraint:
\[
\omega(\tilde{\boldsymbol{u}}, \tilde{\boldsymbol{v}}) = \frac{2\pi q}{p d}, \quad \text{with } d \in \mathbb{Z}.
\]

\color{black}

\section{Output State Structure}
\label{output}
In this section,  we would like to understand the structure of the output states and their dependence on the crosstalk parameter in order to design a decoder that is able to retrieve meaningfully the encoded logical information. Accordingly, 
We analyze the output state resulting from sending two GKP codewords \(|\mu_1\rangle \otimes |\mu_2\rangle\) where $\eta$ is an appropriate rational number for the corresponding code dimensions as highlighted by Eq.~\ref{cond}. We characterize the structure of the output code, derive explicit expressions for the resulting quantum state, and interpret the emergence of a gauge subsystem both algebraically and operationally. Finally, we provide a decoding strategy to retrieve the logical information in the output states. 

\subsection{Code Dimension After crosstalk Coupling}
Despite that the logical information is preserved in the output due to invariance in the stabilizer, the former may be encoded in an embedding space, different from the input space. As such, we provide the following lemma.
\begin{lemma}[Output Code Dimensions]
Let the input GKP codes be $\mathcal{C}_{S_1,d_1}$ and $\mathcal{C}_{S_2, d_2}$. Then after  crosstalk with $\eta$ satisfying Eq.~\ref{cond}, the output state lies in the code basis:
\begin{equation}
\mathcal{C}_{nd_1} \otimes \mathcal{C}_{nd_2}, \quad \text{with } n = q + p d_1 d_2. 
\end{equation}
\end{lemma}

\begin{proof}
Let the input GKP codes $\mathcal{C}_{S_1,d_1}$ and $\mathcal{C}_{S_2,d_2}$ be defined by stabilizer lattices with lattice vectors  $(\mathbf{u}_1, \mathbf{v}_1) $ and  $(\mathbf{u}_2, \mathbf{v}_2)$, satisfying:
\begin{equation}
 \omega(\mathbf{u}_1, \mathbf{v}_1) = 2\pi d_1, \quad \omega(\mathbf{u}_2, \mathbf{v}_2) = 2\pi d_2
\end{equation}
These define the logical area of each code. Under the action of the crosstalk, displacement operators transform as:
\begin{equation}
U_\eta T_1(\alpha) T_2(\beta) U_\eta^\dagger = T_1(\sqrt{\eta} \alpha + \sqrt{1 - \eta} \beta) \cdot T_2(\sqrt{\eta} \beta - \sqrt{1 - \eta} \alpha)
\label{commutation}
\end{equation}
To preserve the GKP code structure under this transformation, the stabilizer displacements of one mode must map into stabilizer displacements of the other. A sufficient condition is:
\begin{equation}
\tilde{\mathbf{u}}_1 = \sqrt{\frac{\eta}{1 - \eta}} \mathbf{u}_2, \quad
\tilde{\mathbf{v}}_1 = \sqrt{\frac{\eta}{1 - \eta}} \mathbf{v}_2
\end{equation}
so that transformed stabilizers from mode 2 cancel displacements in mode 1, where $(\tilde{\boldsymbol{u}},\tilde{\boldsymbol{v}})$ denote the transformed stabilizer generators of mode 1 matching the condition in Eq.~\ref{eq:alpha_requirement}. The  mode remains in a valid GKP code space only if the stabilizer lattice $\Lambda_1' = \mathrm{Span}_{\mathbb{Z}}\{\tilde{\boldsymbol{u}}, \tilde{\boldsymbol{v}}\}$ contains a sublattice $(\boldsymbol{u}_1, \boldsymbol{v}_1)$ such that:
\begin{equation}
\tilde{\boldsymbol{u}}_1 = \frac{r_1}{d_1} \boldsymbol{u}_1 + t_1 \boldsymbol{v}_1, \quad \tilde{\boldsymbol{v}}_1 = \frac{r_2}{d_1} \boldsymbol{v}_1 + t_2 \boldsymbol{u}_1
\label{trans}
\end{equation}
for integers $r_1,r_2t_1, t_2$. That is, the logical shifts must lie in the 2D plane spanned by the transformed stabilizers. From Eq.~\ref{trans}, we note that, if we choose $\alpha=kd_1\tilde{\mathbf{u}}_1=m\mathbf{u}_1$ this yields:
\begin{equation}
\omega_{out,1} = \frac{q \omega(\mathbf{u}_1, \mathbf{v}_1)}{ \eta} = 2\pi n d_1 
\end{equation}
Applying the same logic to mode 2, we note that the output codes must have symplectic areas:
\begin{equation}
\omega_{\text{out},1} = 2\pi n d_1, \quad \omega_{\text{out},2} = 2\pi n d_2,
\end{equation}
corresponding to GKP codes \( \mathcal{C}_{n d_1} \) and \( \mathcal{C}_{n d_2} \). This completes the proof.
\end{proof}

\subsection{General Output State Derivation}
Since we discovered that the output GKP codes live in a larger space scaled by the crosstalk parameter $n$, we have to investigate the structure of the output states in a proper basis. 
\begin{theorem}[General Output State]
Let \( \alpha_1, \alpha_2 \in \mathbb{Z} \) satisfy the Bézout identities:
\[
\alpha_1 q + \beta_1 p d_1 d_2 = 1, \quad \alpha_2 q + \beta_2 p d_1 d_2 = 1.
\]
Then the output of the mode mixing \( U_\eta \) applied to the input state \(|\mu_1\rangle \otimes |\mu_2\rangle\) is:
\[
|\mu_1, \mu_2\rangle = \frac{1}{\sqrt{n}} \sum_{j=0}^{n-1}
\left|\mu_1 \alpha_1 n + j p d_1 d_2 \right\rangle \otimes
\left|\mu_2 \alpha_2 n + j q d_2 \right\rangle,
\]
where the output states belong to \( \mathcal{C}_{n d_1} \otimes \mathcal{C}_{n d_2} \).
\end{theorem}

\begin{proof}
We begin with input GKP codewords \( |\mu_1\rangle \in \mathcal{C}_{d_1} \), \( |\mu_2\rangle \in \mathcal{C}_{d_2} \), which can be formally written as superpositions over lattice displacements:
\[
|\mu_i\rangle = \sum_{s_i \in \mathbb{Z}} T_i(s_i \mathbf{u}_i + \mu_i \mathbf{u}_i/d_i) |0\rangle.
\]
The logical operator \( \hat{X}_i \) corresponds to a displacement \( T_i(\mathbf{u}_i/d_i) \). Applying \( U_\eta \) to the tensor product of these two codewords yields:

\begin{align}
&|\mu_1, \mu_2\rangle = U_\eta \left( |\mu_1\rangle \otimes |\mu_2\rangle \right) \nonumber\\
&= \sum_{s_1, s_2} U_\eta \left( T_1(s_1 \mathbf{u}_1 + \mu_1 \mathbf{u}_1/d_1) \otimes T_2(s_2 \mathbf{u}_2 + \mu_2 \mathbf{u}_2/d_2) \right) |0\rangle|0\rangle.
\end{align}
Using the transformation law under the beamsplitter:
\[
U_\eta T_1(\alpha) T_2(\beta) U_\eta^\dagger = T_1(\sqrt{\eta} \alpha + \sqrt{1 - \eta} \beta) \cdot T_2(\sqrt{\eta} \beta - \sqrt{1 - \eta} \alpha),
\]
each term becomes:
\[
T_1(\lambda_1(s_1,s_2)) \cdot T_2(\lambda_2(s_1,s_2)) |0,0\rangle,
\]
where the displacement vectors are:
\begin{align*}
\lambda_1 &= \sqrt{\eta} (s_1 \mathbf{u}_1 + \mu_1 \mathbf{u}_1/d_1) + \sqrt{1 - \eta} (s_2 \mathbf{u}_2 + \mu_2 \mathbf{u}_2/d_2), \\\\
\lambda_2 &= \sqrt{\eta} (s_2 \mathbf{u}_2 + \mu_2 \mathbf{u}_2/d_2) - \sqrt{1 - \eta} (s_1 \mathbf{u}_1 + \mu_1 \mathbf{u}_1/d_1).
\end{align*}
and $\ket{0,0}=U_{\eta} \ket{0}\ket{0}$. Simplifying the displacement vectors, we noted previously that the stabilizer lattice vector of the outputs are
\begin{equation}
    \mathrm{u}_3= \frac{q_1\mathrm{u}_1}{\eta} \quad  \mathrm{u}_4= \frac{q_2\mathrm{u}_2}{\eta} 
\end{equation}
and that:
\begin{equation}
    \sqrt{1-\eta} \mathbf{u}_2==\frac{1-\eta}{\eta}\frac{q_1}{d_1p_1} \mathbf{u}_1
\end{equation}
we arrive at:
\begin{align}
\lambda_1 &= q_1 \mu_1 \frac{\mathbf{u}_3}{nd1}+p_2d_1\mu_2 \frac{\mathbf{u}_3}{nd1} \\\\
\lambda_2 &= q_2 \mu_2 \frac{\mathbf{u}_4}{nd2}-p_1d_2\mu_1 \frac{\mathbf{u}_4}{nd2} \label{vectors}
\end{align}
with the same relations and with the use of the fact that the output GKP states are considered mod $n$ and mod $n$, we obtain the output vacuum state in terms of the basis vectors $\ket{\mu_3}$ and $\ket{\mu_4}$ as:
\begin{equation}
\ket{0,0}=\sum_{j=0}^{n-1} \ket{jpd_1d_2}\ket{jd_2q}
\end{equation}
Applying the necessary translations with vectors in Eq.~\ref{vectors}, we obtain:
\begin{equation}
\ket{\mu_1,\mu_2}=\sum_{j=0}^{n-1}\ket{q_1\mu_1+p_2d_1\mu_2+jpd_1d_2}\ket{q_2\mu_2-p_1d_2\mu_1+jd_2q}
\end{equation}
considering the states now are left mod $d_1$ and mod $d_2$ respectively, and using the Diophantine equations:
\[
\alpha_i q + \beta_i p d_1 d_2 = 1,
\]
so that \( \mu_i \alpha_i n \equiv \mu_i \mod d_i \)
We get:
\[
|\mu_1, \mu_2\rangle = \frac{1}{\sqrt{n}} \sum_{j=0}^{n-1}
\left|\mu_1 \alpha_1 n + j p d_1 d_2 \right\rangle \otimes
\left|\mu_2 \alpha_2 n + j q d_2 \right\rangle,
\]
which proves the theorem.

This specific structure ensures that:
\begin{itemize}
    \item The logical information \( (\mu_1, \mu_2) \) is encoded into distinct cosets in the enlarged lattices.
    \item The index \( j \) sweeps over the gauge orbit corresponding to a subgroup of the full lattice that does not alter the logical labels.
\end{itemize}

\end{proof}
\color{black}We note that the enlargement of the stabilizer lattices introduces a gauge degree of freedom that increases the size of the decoding space. While the logical subspace remains of dimension \(d_1 \times d_2\), the physical lattice now supports \(n\) gauge cosets, which can increase decoding complexity. In particular, minimum distance decoding must resolve among \(n\) degenerate candidates, potentially raising the complexity from \(\mathcal{O}(d)\) to \(\mathcal{O}(n \cdot d)\) in the worst case. This overhead is mitigated if gauge fixing is applied prior to logical decoding, as will be discussed in Subsection.~\ref{gausge}
\color{black}

\begin{corollary}[Symmetric Case]
In the symmetric setting where \( d_1 = d_2 = d \), \( q_1 = q_2 = 1 \), \( p_1 = p_2 = 1 \), and \( \alpha_1 = \alpha_2 = 1 \), the output simplifies to:
\[
|\mu_1, \mu_2\rangle = \frac{1}{\sqrt{n}} \sum_{j=0}^{n-1}
|\mu_1 n + j d^2\rangle \otimes |\mu_2 n + j d\rangle.
\]
\end{corollary}

\begin{proof}
We consider the general output state from the previous theorem:
\[
|\mu_1, \mu_2\rangle = \frac{1}{\sqrt{n}} \sum_{j=0}^{n-1}
\left| \mu_1 \alpha_1 n + j p_2 d_1 d_2 \right\rangle \otimes
\left| \mu_2 \alpha_2 n + j q_1 d_2 \right\rangle,
\]
and impose the symmetric assumptions:
\[
d_1 = d_2 = d, \quad q_1 = q_2 = 1, \quad p_1 = p_2 = 1, \quad \alpha_1 = \alpha_2 = 1.
\]
Under these conditions, we substitute:
\[
p_2 d_1 d_2 = d^2, \quad q_1 d_2 = d, \quad \alpha_1 n = \alpha_2 n = n.
\]
Then the expression becomes:
\[
|\mu_1, \mu_2\rangle = \frac{1}{\sqrt{n}} \sum_{j=0}^{n-1}
| \mu_1 n + j d^2 \rangle \otimes | \mu_2 n + j d \rangle.
\]
This is a superposition over \( n \) values, with step size \( d \) in both modes, and offsets determined by \( \mu_1, \mu_2 \), each scaled by \( n \), consistent with embedding the logical indices into \( \mathcal{C}_{n d} \otimes \mathcal{C}_{n d} \). Thus, the output state lies in \( \mathcal{C}_{n d} \otimes \mathcal{C}_{n d} \) and has the simplified form claimed.

\end{proof}

\subsection{Gauge Degrees of Freedom and Symmetry}
\label{gausge}
We would like to understand the possibility of decoupling the logical degrees of freedom from the gauge degrees of freedom, to be able to retrieve the logical information witout any gauge coupling.  
\begin{figure}
    \centering
    \includegraphics[width=1\linewidth]{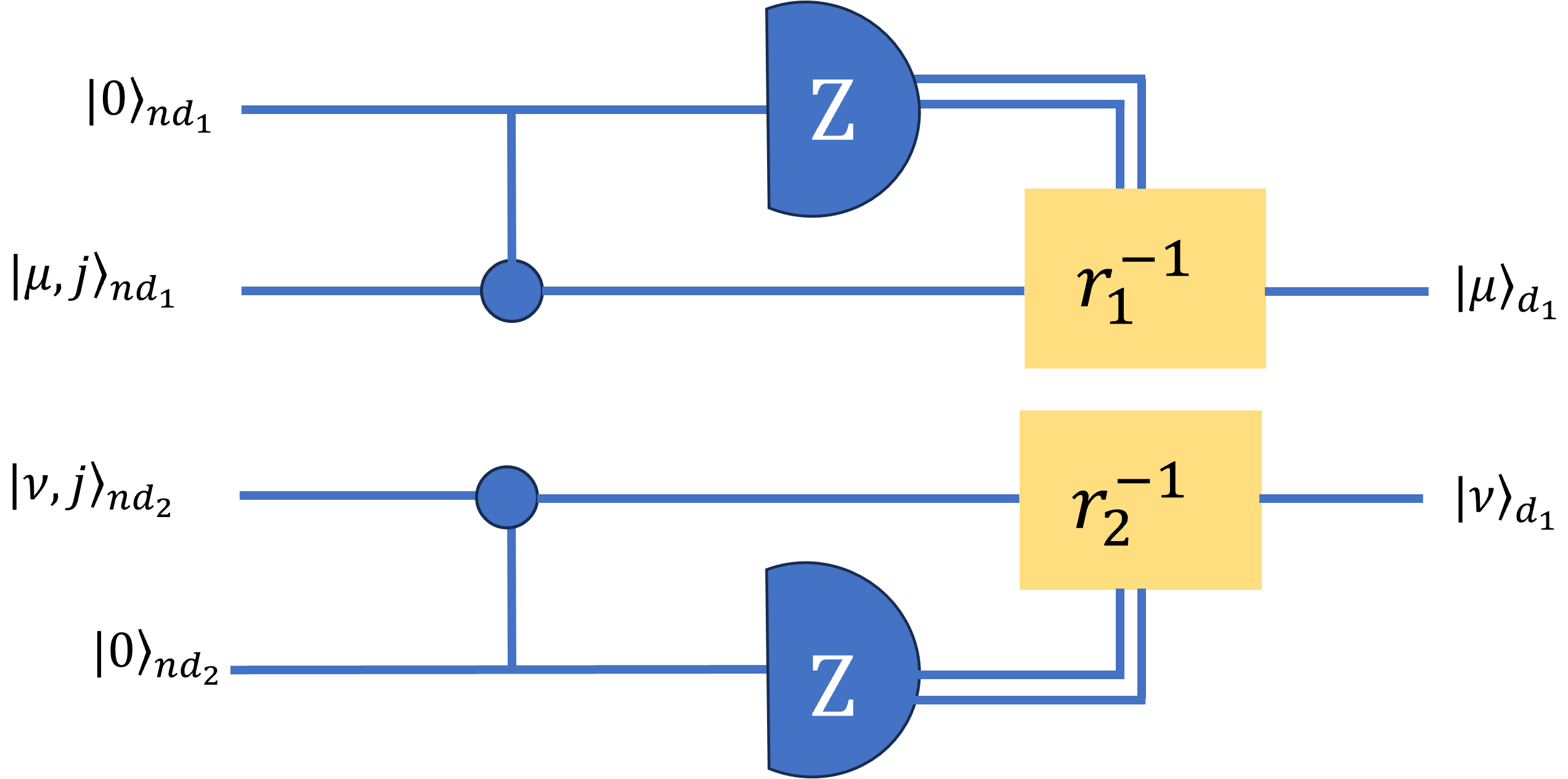}
    \caption{The gauge fixing decoder. A logical Controlled-X gate is applied to the ancillary mode, followed by measurement in the logical $Z$ basis. The gauge parameter is estimated $mod$-n and a logical inverse gate map is applied to the target mode to retrieve the logical state without gauge coupling.}
    \label{decoder}
\end{figure}
\begin{theorem}[Gauge Subsystem Structure]
The output state \(|\mu_1, \mu_2\rangle\) admits a factorization of the form:
\[
|\mu_1, \mu_2\rangle = |\mu_1, \mu_2\rangle_L \otimes |\Phi_n\rangle_G,
\]
where:
\begin{itemize}
  \item \(|\mu_1, \mu_2\rangle_L\) encodes the logical state in \( \mathcal{H}_{d_1} \otimes \mathcal{H}_{d_2} \),
  \item \(|\Phi_n\rangle = \frac{1}{\sqrt{n}} \sum_{j=0}^{n-1} |j\rangle \otimes |j\rangle\) is a maximally entangled state in the gauge subsystem \( \mathcal{H}_n \otimes \mathcal{H}_n \),
  \item The decomposition holds: \( \mathcal{C}_{nd_1} \otimes \mathcal{C}_{nd_2} \cong \mathcal{H}_{d_1} \otimes \mathcal{H}_n \otimes \mathcal{H}_{d_2} \otimes \mathcal{H}_n \).
\end{itemize}
\end{theorem}

\begin{proof}
We start with the general form of the output state derived in the previous theorem:
\[
|\mu_1, \mu_2\rangle = \frac{1}{\sqrt{n}} \sum_{j = 0}^{n - 1}
\left| \mu_1 \alpha_1 n + j p_2 d_1 d_2 \right\rangle \otimes
\left| \mu_2 \alpha_2 n + j q_1 d_2 \right\rangle,
\]
with all terms in the Hilbert space \( \mathcal{C}_{n d_1} \otimes \mathcal{C}_{n d_2} \), which has total dimension \( n d_1 \times n d_2 \).

We aim to show that this state admits a decomposition into a logical component and a gauge-entangled component. To do so, we define a bijective relabeling of the basis of \( \mathcal{C}_{n d_i} \cong \mathcal{H}_{d_i} \otimes \mathcal{H}_n \) via the map:
\[
\mu_i \alpha_i n + j r_i \longmapsto |\mu_i\rangle_L \otimes |j\rangle_G,
\]
where \( r_1 = p_2 d_1 d_2 \), \( r_2 = q_1 d_2 \), and \( j \in \mathbb{Z}_n \). This map is well-defined because the terms \( \mu_i \alpha_i n \) span a sublattice of spacing \( n \), and the shifts \( j r_i \) span cosets within that lattice, labeling the gauge degree of freedom.

We now define a new basis for the enlarged code space:
\[
|\mu_i n + j r_i\rangle \equiv |\mu_i\rangle_L \otimes |j\rangle_G,
\]
so that the full state becomes:
\[
|\mu_1, \mu_2\rangle = \frac{1}{\sqrt{n}} \sum_{j = 0}^{n - 1}
|\mu_1\rangle_L \otimes |j\rangle_G \otimes
|\mu_2\rangle_L \otimes |j\rangle_G.
\]

Regrouping:
\[
|\mu_1, \mu_2\rangle = \left(|\mu_1\rangle_L \otimes |\mu_2\rangle_L \right) \otimes \left( \frac{1}{\sqrt{n}} \sum_{j = 0}^{n - 1} |j\rangle \otimes |j\rangle \right).
\]
We identify the first tensor factor as the logical state \( |\mu_1, \mu_2\rangle_L \in \mathcal{H}_{d_1} \otimes \mathcal{H}_{d_2} \), and the second factor as a maximally entangled Bell-like state \( |\Phi_n\rangle \in \mathcal{H}_n \otimes \mathcal{H}_n \), defined by:
\[
|\Phi_n\rangle = \frac{1}{\sqrt{n}} \sum_{j = 0}^{n - 1} |j\rangle \otimes |j\rangle.
\]

Therefore, the state admits a decomposition:
\[
|\mu_1, \mu_2\rangle = |\mu_1, \mu_2\rangle_L \otimes |\Phi_n\rangle_G,
\]
and the total space decomposes as:
\[
\mathcal{C}_{n d_1} \otimes \mathcal{C}_{n d_2} \cong \mathcal{H}_{d_1} \otimes \mathcal{H}_n \otimes \mathcal{H}_{d_2} \otimes \mathcal{H}_n.
\]
\end{proof}

\begin{lemma}[Symmetry Origin of Gauge Subsystem]
The gauge subsystem arises from the residual symmetry induced by the mode mixing on the stabilizer structure. Define the input stabilizer lattices \( \Lambda_1, \Lambda_2 \) and let \( \Lambda_{\text{logical}} \subset \Lambda_1 \oplus \Lambda_2 \) denote the sublattice that preserves logical equivalence. Then:
\[
G = (\Lambda_1 \oplus \Lambda_2)/\Lambda_{\text{logical}}
\]
is a finite Abelian group of order \( n \) acting on the gauge index \( j \). This symmetry results in maximal entanglement between the gauge subsystems.
\end{lemma}

\begin{proof}
Let \( \Lambda_1, \Lambda_2 \subset \mathbb{R}^2 \) be the input stabilizer lattices of the two GKP codes \( \mathcal{C}_{d_1} \) and \( \mathcal{C}_{d_2} \). These lattices are generated by primitive vectors \( (\mathbf{u}_1, \mathbf{v}_1) \), \( (\mathbf{u}_2, \mathbf{v}_2) \) satisfying symplectic conditions:
\[
\omega(\mathbf{u}_1, \mathbf{v}_1) = 2\pi d_1, \quad \omega(\mathbf{u}_2, \mathbf{v}_2) = 2\pi d_2.
\]

The stabilizer group of the joint system is \( \Lambda_1 \oplus \Lambda_2 \), i.e., the product lattice of displacements acting on both modes.

When the two modes interact through a crosstalk \( U_\eta \), the phase-space coordinates transform linearly:
\[
(\xi_1, \xi_2) \mapsto \left( \sqrt{\eta} \xi_1 + \sqrt{1 - \eta} \xi_2,\, \sqrt{\eta} \xi_2 - \sqrt{1 - \eta} \xi_1 \right).
\]
This transformation mixes the lattices \( \Lambda_1, \Lambda_2 \) into a coupled sublattice in \( \mathbb{R}^4 \), generating new displacement operators acting nontrivially across both modes.
However, not all elements of \( \Lambda_1 \oplus \Lambda_2 \) map to distinct logical states in the output code. There exists a sublattice \( \Lambda_{\text{logical}} \subset \Lambda_1 \oplus \Lambda_2 \) consisting of those displacement combinations that ct trivially on the logical space, i.e., leave logical states invariant up to stabilizer equivalence, and that commute with the logical operators of the output code \( \mathcal{C}_{n d_1} \otimes \mathcal{C}_{n d_2} \). Hence, the quotient group
\[
G = (\Lambda_1 \oplus \Lambda_2)/\Lambda_{\text{logical}}
\]
classifies the inequivalent cosets of logical actions modulo stabilizer redundancy. By construction, \( G \) is a finite Abelian group. Its order is given by:
\[
|G| = \frac{\text{Vol}(\Lambda_{\text{logical}})}{\text{Vol}(\Lambda_1 \oplus \Lambda_2)} = n,
\]
since the logical lattice scales by a factor of \( n \) due to the mode mixing-induced enlargement of the stabilizer unit cell (as previously shown in the output code dimension lemma).
This finite group \( G \) acts transitively on the label \( j \) that indexes the gauge orbit in the output state:
\[
|\mu_1, \mu_2\rangle = \frac{1}{\sqrt{n}} \sum_{j=0}^{n-1} |\mu_1^{(j)}\rangle \otimes |\mu_2^{(j)}\rangle.
\]
Thus, the existence of the gauge subsystem arises from this residual symmetry: different elements of the quotient \( G \) correspond to gauge-equivalent configurations of logical information that are indistinguishable by any logical observable. The entanglement of the gauge subsystems reflects the nontrivial correlation between these symmetry orbits across both modes.

\end{proof}

\subsection{Gauge Fixing and Logical Recovery}

Having the output state depending on the gauge degrees of freedom, does not allow for perfect recovery of the logical information unless the gauge is fixed and \textcolor{black}{then is collapsed} to one of its orbits. In what follows we provide an existence argument of such gauge fixing decoder.

\begin{theorem}[Gauge Fixing Procedure]
Let \(|\psi_{\mu_1, \mu_2}\rangle\) be the output state as above. Then there exists a unitary \( U_{\mathrm{gauge-fix}} \) acting on \( \mathcal{C}_{n d_1} \otimes \mathcal{C}_{n d_2} \) such that:
\[
U_{\mathrm{gauge-fix}} |\psi_{\mu_1, \mu_2}\rangle = |\mu_1\rangle \otimes |\mu_2\rangle \otimes |0\rangle_G,
\]
where \(|\mu_1\rangle, |\mu_2\rangle\) are logical states in \( \mathcal{C}_{d_1}, \mathcal{C}_{d_2} \), and \(|0\rangle_G\) is a fixed state in the gauge register. The unitary \( U_{\mathrm{gauge-fix}} \) is constructible from modular arithmetic operations and GKP logical Clifford gates (e.g., modular multiplication, modular subtraction, controlled shift).
\end{theorem}

\begin{proof}
We begin with the output state of the form derived earlier:
\[
|\psi_{\mu_1, \mu_2}\rangle = \frac{1}{\sqrt{n}} \sum_{j = 0}^{n - 1}
\left|\mu_1 \alpha_1 n + j r_1 \right\rangle \otimes
\left|\mu_2 \alpha_2 n + j r_2 \right\rangle,
\]
where \( r_1 = p_2 d_1 d_2 \), \( r_2 = q_1 d_2 \), and \( \alpha_i \in \mathbb{Z} \) satisfy the Bézout identities. This state lies in the Hilbert space \( \mathcal{C}_{n d_1} \otimes \mathcal{C}_{n d_2} \), which we know decomposes as:
\[
\mathcal{C}_{n d_1} \otimes \mathcal{C}_{n d_2} \cong \mathcal{H}_{d_1} \otimes \mathcal{H}_n \otimes \mathcal{H}_{d_2} \otimes \mathcal{H}_n.
\]
We aim to apply a unitary transformation \( U_{\mathrm{gauge-fix}} \) that maps the entangled gauge part into a separable form and re-expresses the state as:
\[
|\mu_1\rangle \otimes |\mu_2\rangle \otimes |0\rangle_G.
\]
To construct this unitary, observe that the terms \( \mu_1 \alpha_1 n + j r_1 \) index states of the form:
\[
|\mu_1\rangle_L \otimes |j\rangle_G,
\]
under the identification:
\[
|\mu_1\rangle_L = |\mu_1 \alpha_1 n \mod n d_1\rangle, \quad |j\rangle_G = |j r_1 \mod n d_1\rangle.
\]
Hence, the full state is:
\[
|\psi_{\mu_1, \mu_2}\rangle = |\mu_1\rangle_L \otimes |\mu_2\rangle_L \otimes \left( \frac{1}{\sqrt{n}} \sum_{j=0}^{n - 1} |j\rangle_G \otimes |j\rangle_G \right).
\]
Define the unitary \( U_{\mathrm{gauge-fix}} \) as follows:
\begin{itemize}
  \item Apply modular inverse multiplication: \( j \mapsto r_i^{-1} j \mod n \), to bring gauge indices into canonical form \( 0,1,\dots,n-1 \),
  \item Use a controlled modular subtraction to align both gauge indices,
  \item Shift the gauge register into a basis where \( |j\rangle \otimes |j\rangle \mapsto |0\rangle \otimes |j\rangle \),
  \item Then discard (or reinitialize) the second register.
\end{itemize}
All these operations can be constructed from GKP logical Clifford unitaries:
\begin{itemize}
  \item Modular multiplication and inverses can be implemented via shift and scaling gates,
  \item Controlled modular subtraction is a Clifford operation on GKP codes,
  \item Basis change from entangled to separable form is equivalent to a Bell-to-product basis transformation.
\end{itemize}
After applying \( U_{\mathrm{gauge-fix}} \), we obtain:
\[
U_{\mathrm{gauge-fix}} |\psi_{\mu_1, \mu_2}\rangle = |\mu_1\rangle \otimes |\mu_2\rangle \otimes |0\rangle_G,
\]
with \( |0\rangle_G \) a fixed state (e.g., the zero state in the canonical basis of \( \mathcal{H}_n \)).

\end{proof}

\subsection{Operational Construction of the Decoder}
\begin{figure}
    \centering
    \includegraphics[width=1\linewidth]{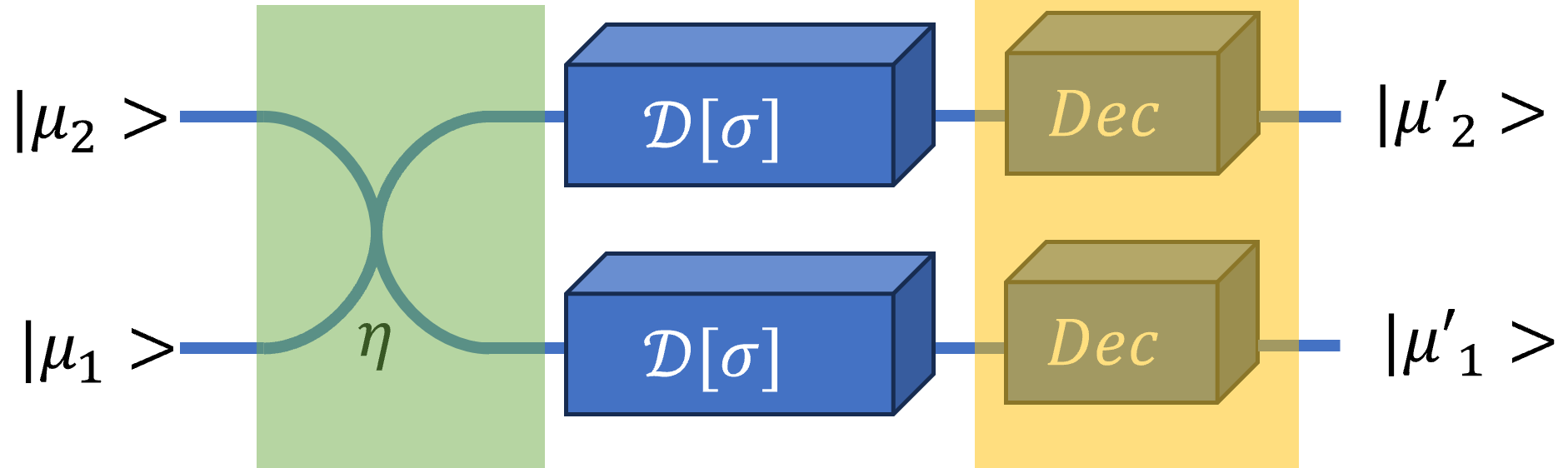}
    \caption{The scheme considered for the simulation. We consider transmitted GKP encoded EPR shares $\ket{\mu_1}$ and  $\ket{\mu_2}$ affected by crosstalk modeled by a beam splitter of transmitivity $\eta$. The two modes then are affected by displacement gaussian noise. Upon reception, an ideal \textcolor{black}{disjoint} decoding scheme as described in the text, tracing out the gauge degrees of freedom, is applied}
    \label{fig:enter-label}
\end{figure}
We give an operational construction of the decoder for gauge fixing the output state to recover the logical state:

\begin{enumerate}

\item \textbf{Ancilla Preparation:}  
Initialize two ancillary GKP modes, each with Hilbert space dimension \( n d_1 \) and \( n d_2 \), respectively, in the logical vacuum state \( |0\rangle \).

\item \textbf{Entangling Coupling:}  
Apply modular controlled-X (\textsf{CX}) gates from each output mode to its respective ancilla. These gates act as \( \textsf{CX} : |x\rangle|0\rangle \mapsto |x\rangle|x\rangle \mod n d_i \), entangling the mode and ancilla. The shared gauge index \( j \) is thus encoded redundantly in both ancillas.

\item \textbf{Measurement:}  
Measure the ancillas in the computational basis, yielding outcomes:
\[
x = \mu_1 \alpha_1 n + j r_1, \quad y = \mu_2 \alpha_2 n + j r_2
\]
modulo \( n d_1 \) and \( n d_2 \), respectively. These results reveal the value of the gauge index \( j \mod n \) via modular inversion:
\[
j = x \cdot r_1^{-1} \mod n = y \cdot r_2^{-1} \mod n
\]
since \( \gcd(r_1, n) = \gcd(r_2, n) = 1 \) by construction, hence this is a physical unitary gate.

\item \textbf{Feedback Correction:}  
Apply modular displacement gates (logical \( \bar{X} \) or \( \bar{Z} \)) to the output modes, conditioned on the measured value of \( j \). These correct for the gauge entanglement, mapping:
\[
|\mu_1 \alpha_1 n + j r_1\rangle \mapsto |\mu_1 \alpha_1 n\rangle, \quad
|\mu_2 \alpha_2 n + j r_2\rangle \mapsto |\mu_2 \alpha_2 n\rangle
\]
This is equivalent to applying a modular shift \( -j r_i \mod n d_i \), which is a GKP logical Clifford operation.

\item \textbf{Gauge Reset:}  
After correction, the ancillas are disentangled from the logical system. They can be reinitialized to \( |0\rangle \), discarded, or traced out. The logical state is now in canonical form.

\end{enumerate}
This sequence constitutes an operational realization of the unitary gauge-fixing map \( U_{\mathrm{gauge-fix}} \), acting unitarily and reversibly on the joint system, up to classical feedforward corrections based on the gauge index \( j \).
\color{black}
It is worth-noting that the above gauge-fixing protocol assumes that the receiver knows the gauge dimension \( n \), which depends on the rational structure of the average mode-mixing strength  \(\eta\), but no information on the stochastic crosstalk behavior. In scenarios without channel state information (CSI), this value is unknown, and modular inversion or correction cannot be deterministically applied. Therefore, this decoding method is valid in the \emph{with-CSI} regime. Extending gauge correction to the \emph{no-CSI} setting would require adaptive or statistical decoding strategies beyond the scope of this paper.
\color{black}

\subsection{Numerical Simulations and Fidelity Bounds}
\begin{figure}
    \centering
    \includegraphics[width=1\linewidth]{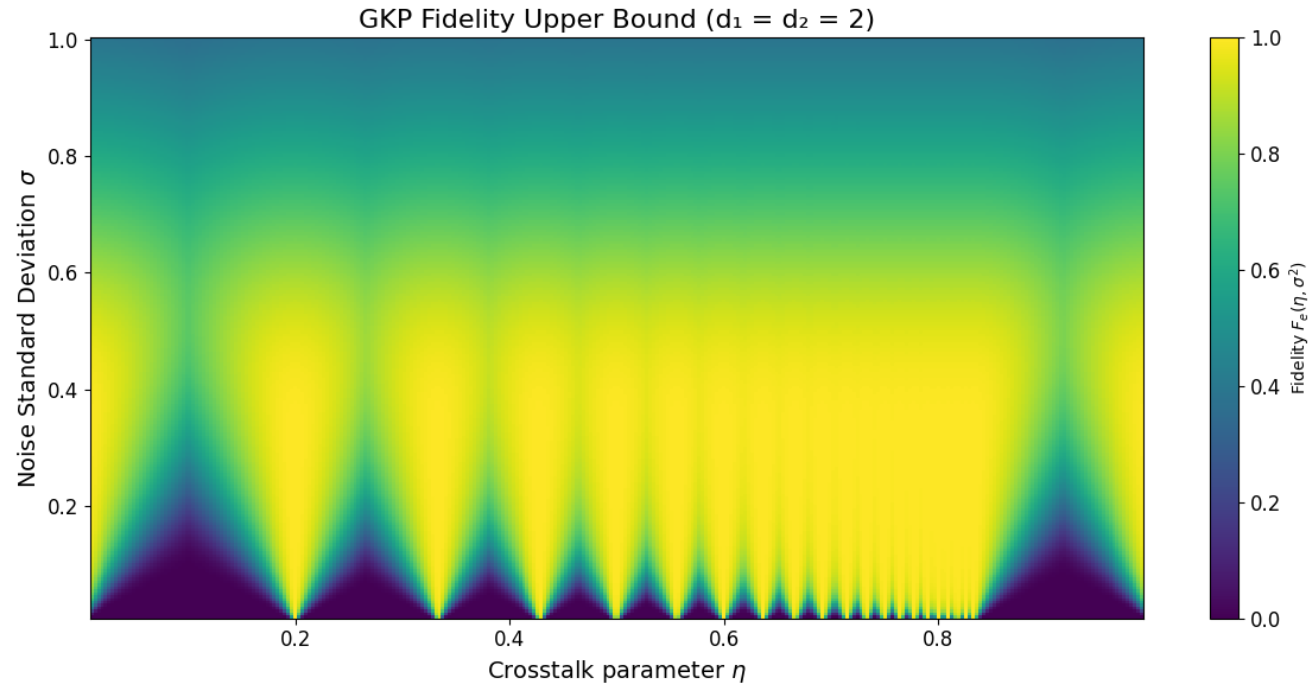}
    \caption{Upper bound on the entangleemnt fidelity for GKP-encoded EPR pairs transmission over mode mixing crosstalk with transmissivity $\eta$ and dipslacement Gaussian noise of variance $\sigma$. The plot assumes square GKP codes and dimensions $d_1=d_2=2$ and $\sigma_c=0.4$ }
    \label{Fidelity}
\end{figure}

\begin{figure}
    \centering
    \includegraphics[width=1\linewidth]{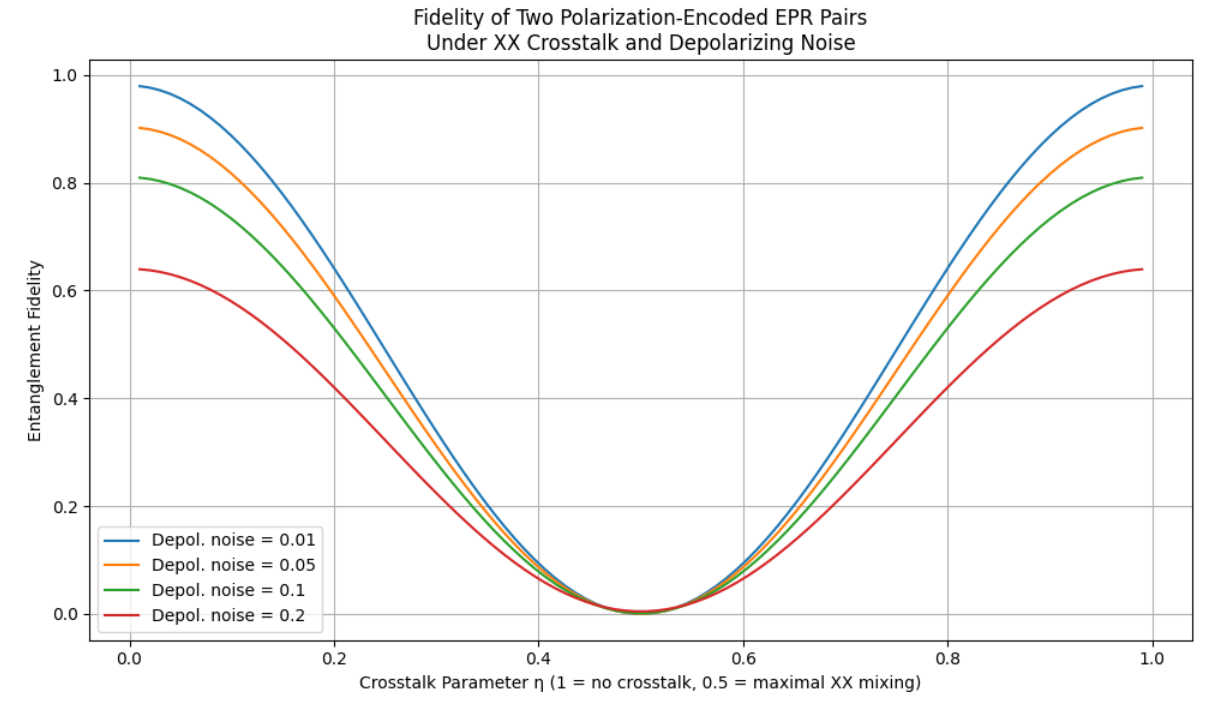}
    \caption{Entanglement fidelity of two polarization encoded EPR pairs under crosstalk and depolarizing noise. The tranmitted shares of the EPR pairs interact via a coherent $XX$ coupling with strength controlled by parameter $\eta$}
    \label{DV fidelity}
\end{figure}
In this subsection, we illustrate the practical implications of the multiplexing structure and rational transmissivity conditions by computing fidelity upper bounds for finite-energy GKP states under cross-talk and Gaussian noise. We consider a scenario where two independent GKP encoded EPR pairs  prepared in parallel: one between parties A and A', and the other between B and B'. Each half-pair, A' and B', is  transmitted through a shared channel with passive crosstalk with parameter \( \eta = \frac{q}{q + p d^2} \). The other halves, A and B, are held locally.

The joint state before transmission is:
\[
|\Psi\rangle = \frac{1}{2} \sum_{\mu, \nu = 0}^1 |\mu\rangle_A \otimes |\mu\rangle_{A'} \otimes |\nu\rangle_B \otimes |\nu\rangle_{B'}.
\]

After encoding and crosstalk interaction, the transmitted shares transform as:
\[
|\mu, \nu\rangle = \frac{1}{\sqrt{n}} \sum_{j = 0}^{n - 1} |n\mu + j r_1\rangle_{A'} \otimes |n\nu + j r_2\rangle_{B'},
\]
where \( r_1 = p_2 d_1 d_2 \), \( r_2 = q_1 d_2 \), and \( \alpha_1 = \alpha_2 = 1 \).
For example, with \( d_1 = d_2 = 2 \), \( \eta = \frac{1}{5} \), \( n = 5 \), the output becomes:
\[
|\Psi\rangle_{\text{after}} = \frac{1}{2\sqrt{5}} \sum_{\mu, \nu = 0}^{1} \sum_{j = 0}^{4}
|\mu\rangle_A \otimes |\nu\rangle_B \otimes |5\mu + 4j\rangle_{A'} \otimes |5\nu + 2j\rangle_{B'}.
\]

\subsubsection{Displacement Noise Model}
We assume a standard GKP error model where noise acts via random phase-space displacements drawn from an isotropic Gaussian distribution. A single-mode displacement noise channel with variance \( \sigma^2 \) is modeled as:

\begin{equation}
    \rho \rightarrow \int P(\alpha) D(\alpha)\rho D^\dagger(\alpha) d\alpha
\end{equation}
with the Gaussian distribution for  the displacement channel  given by:

\begin{equation}
P(\alpha) = \frac{1}{\sqrt{2\pi \sigma^2}} e^{-\frac{\alpha^2}{2\sigma^2}}
\end{equation}

This model captures thermal or stochastic noise affecting position and momentum symmetrically, and is appropriate for modeling finite-energy GKP states under Gaussian channels.

\subsubsection{Derivation of Ideal Decoder Fidelity}
Assuming exact rational transmissivity (i.e., perfect lattice matching), the decoding fidelity for a single GKP mode under Gaussian noise is the probability that the displacement remains within the Voronoi cell of the nearest lattice point. For square GKP lattices, this leads to:
\[
F_{\text{single}}(\sigma, d) = \text{erf}\left( \sqrt{\frac{\pi}{2 d \sigma^2}} \right)^2.
\]
For a two-mode product GKP code \( \mathcal{C}_{d_1} \otimes \mathcal{C}_{d_2} \), the joint fidelity is simply the square of the individual success probabilities:
\[
F_{\text{ideal}}(\sigma; d_1, d_2) = \left[ \text{erf} \left( \sqrt{\frac{\pi}{2 d_1 \sigma^2}} \right) \cdot \text{erf} \left( \sqrt{\frac{\pi}{2 d_2 \sigma^2}} \right) \right]^2.
\]
This expression assumes ideal (infinite-energy) GKP decoding in the absence of lattice mismatches.

 In Fig.~\ref{Fidelity}, the fidelity landscape as function of \( \eta \) and noise Level \( \sigma \) is given. We compute an upper bound on the entanglement fidelity \( F_e(\eta, \sigma^2) \) of the GKP code under Gaussian noise and cross-talk, where rational transmissivity points of the form \( \eta = \frac{q}{q + p d_1 d_2} \) are expected to support exact code embedding. The effective fidelity is bounded by:
\[
F_e(\eta, \sigma^2) \leq \max_{\eta_i} \exp\left( -\frac{\delta(\eta, \eta_i)^2}{2 \sigma^2} \right) F_{\text{ideal}}(\sigma),
\]
where \( \delta(\eta, \eta_i) = |\eta - \eta_i| L \) is the geometric mismatch between the actual \( \eta \) and rational alignment points \( \eta_i \), scaled by a lattice parameter \( L \). The color plot is produced across a grid of \( \eta \) and \( \sigma \), with bright bands indicating regions of high-fidelity transmission centered around rational \( \eta \)-values, with exponential decay in the neighberhood of each peak. In comparison to non-encoded EPR pairs in Fig.~\ref{DV fidelity}, we notice that bosonic encoding of EPR pairs via GKP codes improve drastically the fidelity of entanglement multiplexing in the presence of mode mixing crosstalk (XX coupling between polarization encoded EPR's), especially when the mixing is maximal around $\eta=0.5$. 

 Fig.~\ref{multiplexing}, a simulation of the trade-off between code dimension and fidelity was carried. We fix a rational alignment \( \eta = \frac{q}{q + p d^2} \) with \( d = d_1 = d_2 \), and explore the decay in entanglement fidelity \( F_e \) as \( d \) increases. The fidelity decays with increasing code dimension due to reduced spacing between lattice points and increased sensitivity to noise.  Each curve corresponds to a fixed noise level \( \sigma \), and rational transmissivity points \( \eta(d) \) are annotated for interpretation. Polynomial fits highlight smooth decay behavior, illustrating the trade-off between \textcolor{black}{code rate} (larger \( d \)) and achievable fidelity.

These simulations confirm the analytical expectations that high-fidelity code transmission occurs near rational values of \( \eta \), consistent with the gauge-based embedding theory. The numeric exploration also shows how fidelity decreases with code size due to practical effects of Gaussian noise, emphasizing the role of rational \( \eta \) tuning and finite squeezing in real devices. We note that for both simulations, the ideal decoder from the previous section was considered, which is not optimal in the presence of displacement gaussian noise.
\begin{figure}
    \centering
    \includegraphics[width=1\linewidth]{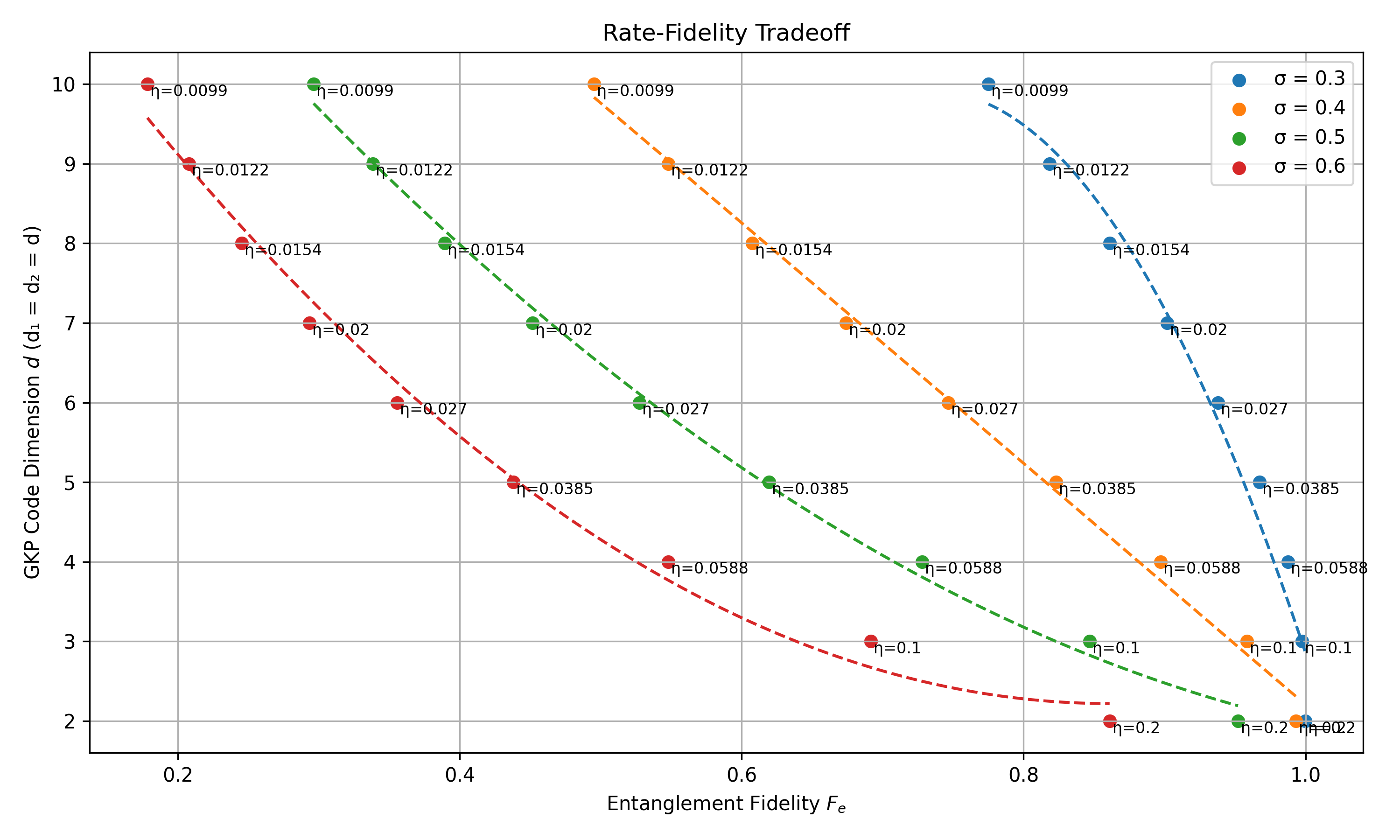}
    \caption{The rate-fidelity \textcolor{black}{tradeoff}}
    \label{multiplexing}
\end{figure}
\section{Conclusions and Future Directions}
\color{black}
\label{conclusions}
We have presented a framework for mitigating crosstalk in multiplexed quantum communication systems using GKP encoding. By embedding DV quantum information into CV modes, we identified precise lattice-matching conditions under which logical information can be preserved despite mode-mixing interactions. A key result is the identification of rational transmissivity values that allow the joint stabilizer structure to be preserved, enabling crosstalk-immune communication. Our work establishes a foundation for crosstalk-aware encoding in quantum networks based on CV hardware and opens up new directions in entanglement distribution through structured code design. Importantly, the principles and techniques developed here for crosstalk-resilient multiplexing are directly applicable to the development of robust Quantum Multiple-Input Multiple-Output (MIMO) communication systems, which are critical for future high-capacity quantum networks.

Additionally, we proposed an explicit decoder that leverages the symplectic structure of the output state, decomposes it into logical and gauge components, and applies a modular arithmetic-based recovery procedure. While this decoder is effective in idealized conditions, it is suboptimal in the presence of Gaussian displacement noise and stochasticity of the crosstalk. Future work will investigate optimized decoding strategies that account for physical noise models, including error rates and channel-specific distortions, particularly within the context of Quantum MIMO channels. In particular, we aim to adapt the decoder to scenarios with no CSI, where the average mode mixing  $\eta$ is unknown. Such robust decoding would enable practical implementation of GKP-based crosstalk mitigation in dynamic network environments and pave the way for truly scalable Quantum MIMO communications. We also plan to explore the direct application of these crosstalk mitigation strategies in full Quantum MIMO architectures, including the optimization of encoding and decoding schemes for multiple spatial modes.
\color{black}

\bibliographystyle{IEEEtran}
\bibliography{ref}
\end{document}